\documentclass[11pt]{article}

\usepackage[a4paper,margin=28mm]{geometry}
\usepackage{setspace}
\usepackage{graphicx}
\usepackage{mathtools}
\usepackage{amsthm}
\usepackage{amssymb}
\usepackage[authoryear,round]{natbib}
\usepackage{tikz}
\usepackage{pgfplots}
\pgfplotsset{compat=1.18}

\usepackage{booktabs}
\usepackage{multirow}
\usepackage{enumitem}
\usepackage{caption}
\usepackage[hidelinks]{hyperref}
\usepackage{authblk}

\numberwithin{equation}{section}

\newtheorem{theorem}{Theorem}[section]

\newtheorem{lemma}[theorem]{Lemma}
\newtheorem{proposition}[theorem]{Proposition}

\theoremstyle{definition}
\newtheorem{definition}[theorem]{Definition}
\newtheorem{assumption}[theorem]{Assumption}

\theoremstyle{remark}
\newtheorem{remark}[theorem]{Remark}

\title{Generalized Ridge Refitting for the Lasso and Prediction Improvement Bounds}
\author[1]{Guo Liu}
\affil[1]{Department of Pure and Applied Mathematics, Waseda University}
\date{September 10, 2026}

\begin{document}
	
	\maketitle

	\begin{abstract}  
		We study a class of Lasso based estimators obtained by applying a quadratic correction on the Lasso equicorrelation set. The penalty matrix determines both the magnitude and geometry of the correction and contains, among other cases, the isotropic Lasso--Ridge correction, least squares refitting, Gram proportional interpolation between the Lasso and least squares, and coordinate specific penalties. We first derive a closed form representation and isolate the positive gain component of the resulting prediction improvement. We then control the remaining stochastic linear term in expectation by localizing the random signed equicorrelation model around a deterministic reference support. This yields a finite sample expectation bound that explicitly accounts for the randomness induced by Lasso model selection. The resulting decomposition provides a unified framework for understanding when Lasso based quadratic corrections can improve prediction.
	\end{abstract}
	
	\section{Introduction}  
	\label{sec:introduction}  
	
	The Lasso is a central method for prediction and variable selection in high dimensional linear regression \citep{tibshirani1996regression}, but its \(\ell_1\) penalty introduces shrinkage bias. A natural way to reduce this bias is to use the Lasso as a first stage model selection procedure and then perform a second refitting step on the selected variables; see, for example, \citet{chzhen2019lasso}. Here we retain the Lasso fit and its equicorrelation set in the first stage and then apply a structured correction on the selected coordinates. The isotropic Lasso--Ridge correction regularizes this adjustment by a ridge penalty. We replace that isotropic penalty by the general quadratic form  
	\begin{align*}  
		\frac12\|A\delta_E\|_2^2  
		=\frac12\delta_E^\top A^\top A\delta_E,  
	\end{align*}  
	where \(A\) is a penalty matrix and \(E\) is the Lasso equicorrelation set. The resulting estimator is  
	\begin{align*}  
		\hat{\beta}_A=\hat{\beta}_{\mathrm L}+\hat{\delta}_A,  
	\end{align*}  
	with correction  
	\begin{align*}  
		\hat{\delta}_{A,E}  
		=(\Sigma_{n,E}+A^\top A)^{-1}\lambda_{\mathrm L}s,  
		\qquad  
		\hat{\delta}_{A,-E}=0.  
	\end{align*}  
	Thus the first stage Lasso estimator and equicorrelation set remain unchanged, while \(A^\top A\) controls the direction and strength of the second stage adjustment.  
	
	The quadratic formulation unifies a broad family of Lasso based procedures. It includes the original isotropic Lasso--Ridge correction, least squares refitting on the Lasso equicorrelation set, Gram proportional corrections that interpolate between the Lasso and least squares, and coordinate specific or low rank penalties. It also clarifies the local relationship with the relaxed Lasso on portions of the relaxation path where the selected support and sign pattern are unchanged.  
	
	Our main theoretical objective is to characterize the prediction improvement over the Lasso. For the realized data, define  
	\begin{align*}  
		\Delta(\hat{\beta}_{\mathrm L},\hat{\beta}_A)  
		\coloneqq  
		\frac1n\|X\hat{\beta}_{\mathrm L}-X\beta_0\|_2^2  
		-\frac1n\|X\hat{\beta}_A-X\beta_0\|_2^2.  
	\end{align*}  
	We decompose this quantity into a positive gain term generated by correcting Lasso shrinkage and a noise dependent linear form. Under general position, the gain term is strictly positive whenever the realized equicorrelation set is nonempty, and it admits an explicit lower bound when the penalty is controlled relative to the restricted Gram matrix. The remaining difficulty is that the noise term is evaluated at the data dependent equicorrelation set and sign vector.  
	
	To address this selection induced dependence, we introduce the deterministic class of signed equicorrelation models attainable by the Lasso and localize the realized model around a reference support. A maximal inequality on events of small probability then yields a finite sample expectation bound governed by three types of quantities: the probability and combinatorial complexity of localized signed model mismatch, the spectral amplification of the correction on the localized and reference models, and the global complexity and amplification weighted by the probability of leaving the localization event. The analysis does not require exact support recovery. Under suitable sparse design conditions the local spectral factors remain bounded, while the global amplification remains explicit and can also be controlled directly by the quadratic penalty.  
	
	The rest of the paper is organized as follows. Section~\ref{sec:setup} introduces notation and the Lasso quantities used throughout. Sections~\ref{sec:generalized-lasso-ridge} and~\ref{sec:basic-properties} develop the generalized correction and its principal special cases. Section~\ref{sec:prediction-improvement} isolates the positive prediction gain and provides a lower bound. Section~\ref{sec:refined-localization-lasso-correction} controls the noise dependent term by localization and relates the required assumptions to sparse eigenvalue conditions and existing Lasso localization results.
	
	\section{Setup and Notation}
	\label{sec:setup}
	
	For a positive integer \(p\), let \([p]\coloneqq\{1,\ldots,p\}\), and denote the cardinality of a finite set \(S\) by \(|S|\).
	
	For a vector \(x\), define
	\begin{align*}
		\|x\|_1 &\coloneqq \sum_i |x_i|, \\
		\|x\|_2 &\coloneqq (\sum_i x_i^2)^{1/2}.
	\end{align*}
	For a matrix \(B=(B_{ij})\), define
	\begin{align*}
		\|B\|_\infty &\coloneqq \max_i \sum_j |B_{ij}|, \\
		\|B\|_2 &\coloneqq \sup_{\|x\|_2=1}\|Bx\|_2.
	\end{align*}
	
	For any \(S\subseteq[p]\), \(X_S\) denotes the submatrix of \(X\) consisting of the columns indexed by \(S\). For a vector \(v\in\mathbb{R}^p\), \(v_S\) denotes the corresponding subvector, and \(v_{-S}\) denotes the subvector indexed by the complement \(S^c=[p]\setminus S\).
	
	We adopt the standard conventions for zero dimensional vectors and matrices. In particular, zero dimensional subvectors and submatrices are well defined, and the \(0\times0\) matrix is treated as a valid square matrix whose inverse is itself. Under these conventions, the formulas below remain valid when the relevant index set is empty.
	
	We consider the linear model
	\begin{align*}
		y=X\beta_0+\varepsilon,
	\end{align*}
	where
	\begin{align*}
		y\in\mathbb{R}^n,\qquad X\in\mathbb{R}^{n\times p},\qquad \beta_0\in\mathbb{R}^p,\qquad \varepsilon\in\mathbb{R}^n.
	\end{align*}
	Throughout the paper, the design matrix \(X\) is treated as fixed. Whenever probabilities or expectations are invoked, the Lasso tuning parameter is deterministic and randomness is taken with respect to the noise \(\varepsilon\). Unless stated otherwise, no distributional assumption is imposed on \(\varepsilon\).
	
	\begin{definition}[Lasso estimator and associated quantities]
		\label{def:lasso}
		Let \(\lambda_{\mathrm L}>0\). The Lasso estimator is defined by
		\begin{align*}
			\hat{\beta}_{\mathrm L}
			\in
			\underset{\beta\in\mathbb{R}^p}{\arg\min}
			\left\{
			\frac{1}{2n}\|y-X\beta\|_2^2
			+
			\lambda_{\mathrm L}\|\beta\|_1
			\right\}.
		\end{align*}
		Associated with \(\hat{\beta}_{\mathrm L}\), define the \emph{equicorrelation set}
		\begin{align*}
			E
			\coloneqq
			\left\{
			j\in[p]:
			\left|
			\frac{1}{n}X_j^\top
			(y-X\hat{\beta}_{\mathrm L})
			\right|
			=
			\lambda_{\mathrm L}
			\right\},
		\end{align*}
		the \emph{restricted Gram matrix}
		\begin{align*}
			\Sigma_{n,E}
			\coloneqq
			\frac{1}{n}X_E^\top X_E,
		\end{align*}
		and the \emph{equicorrelation sign vector} \(s\in\{-1,1\}^{|E|}\),
		\begin{align*}
			s
			\coloneqq
			\operatorname{sign}
			(
			X_E^\top
			(y-X\hat{\beta}_{\mathrm L})
			).
		\end{align*}
	\end{definition}
	
	\begin{proposition}[Lasso KKT conditions]
		\label{prop:lasso-kkt}
		The Karush--Kuhn--Tucker conditions for the Lasso give
		\begin{align*}
			\frac{1}{n}X^\top
			(y-X\hat{\beta}_{\mathrm L})
			=
			\lambda_{\mathrm L}\bar z,
		\end{align*}
		where \(\bar z\in\partial\|\hat{\beta}_{\mathrm L}\|_1\), with
		\begin{align*}
			\bar z_j
			=
			\begin{cases}
				\operatorname{sign}((\hat{\beta}_{\mathrm L})_j),
				& \text{if }(\hat{\beta}_{\mathrm L})_j\neq0, \\[4pt]
				u_j\in[-1,1],
				& \text{if }(\hat{\beta}_{\mathrm L})_j=0.
			\end{cases}
		\end{align*}
		In particular, restricting the KKT conditions to the equicorrelation set gives
		\begin{align*}
			\frac{1}{n}X_E^\top
			(y-X\hat{\beta}_{\mathrm L})
			=
			\lambda_{\mathrm L}s.
		\end{align*}
		Moreover,
		\begin{align*}
			\hat{\beta}_{\mathrm L,-E}=0.
		\end{align*}
	\end{proposition}
	
	We impose the following standard assumption.
	
	\begin{assumption}[General position]
		\label{ass:general-position}
		The columns of \(X\) are in general position. More precisely, for every integer \(k<\min\{n,p\}\), no \(k\)-dimensional affine subspace of \(\mathbb{R}^n\) contains more than \(k+1\) elements of the signed set
		\begin{align*}
			\{\pm X_1,\ldots,\pm X_p\},
		\end{align*}
		excluding antipodal pairs.
	\end{assumption}
	
	The condition of general position is satisfied with probability one for random design models with an absolutely continuous joint distribution. Under Assumption~\ref{ass:general-position}, the Lasso solution is unique. Moreover, for fixed \(\lambda_{\mathrm L}>0\), its support coincides with the equicorrelation set \(E\) for almost every \(y\in\mathbb{R}^n\).
	
	Assumption~\ref{ass:general-position} is imposed throughout unless stated otherwise.
	
	\begin{proposition}[Lasso closed form representation]
		\label{prop:lasso-closed-form}
		Under Assumption~\ref{ass:general-position}, the Lasso solution is unique and
		\begin{align*}
			\operatorname{null}(X_E)=\{0\};
		\end{align*}
		see \citet{tibshirani2013lasso}. Consequently, \(X_E\) has full column rank and, whenever \(|E|>0\), \(\Sigma_{n,E}\) is positive definite. Moreover,
		\begin{align*}
			\hat{\beta}_{\mathrm L,E}
			&=
			(X_E^\top X_E)^{-1}
			(
			X_E^\top y
			-
			n\lambda_{\mathrm L}s
			) \\
			&=
			\Sigma_{n,E}^{-1}
			\left(
			\frac{1}{n}X_E^\top y
			-
			\lambda_{\mathrm L}s
			\right),
			\qquad
			\hat{\beta}_{\mathrm L,-E}=0.
		\end{align*}
	\end{proposition}
	
	\section{The Generalized Lasso--Ridge Estimator}
	\label{sec:generalized-lasso-ridge}
	
	Let \(q\) be a nonnegative integer and \(A\in\mathbb{R}^{q\times|E|}\). In the optimization problem below, \(A\) is treated as fixed with respect to the correction vector, although it may be deterministic or constructed from the observed data. For the expectation analysis in Section~\ref{sec:refined-localization-lasso-correction}, we specialize this general construction to a deterministic family of penalty matrices indexed by the candidate support, so that the realized matrix is random only through the selected equicorrelation set.
	
	\begin{definition}[Generalized Lasso--Ridge estimator]
		\label{def:generalized-lasso-ridge}
		The generalized Lasso--Ridge estimator associated with \(A\) is defined by
		\begin{align*}
			\hat{\beta}_A =\hat{\beta}_{\mathrm L}+\hat{\delta}_A,
		\end{align*}
		where
		\begin{align*}
			\hat{\delta}_A \in\underset{ \substack{ \delta\in\mathbb{R}^p \\
					\delta_{-E}=0 } }{\arg\min} \left\{\frac{1}{2n}\|y-X\hat{\beta}_{\mathrm L}-X\delta \|_2^2+\frac{1}{2}\|A\delta_E\|_2^2\right\}.
		\end{align*}
	\end{definition}
	
	Since \(\hat{\beta}_{\mathrm L,-E}=0\), Definition~\ref{def:generalized-lasso-ridge} can equivalently be written as
	\begin{align*}
		\hat{\beta}_A \in\underset{ \substack{ \beta\in\mathbb{R}^p \\
				\beta_{-E}=0 } }{\arg\min} \left\{\frac{1}{2n}\|y-X\beta\|_2^2+\frac{1}{2} \|A(\beta_E-\hat{\beta}_{\mathrm L,E}) \|_2^2\right\}.
	\end{align*}
	
	The quadratic penalty can be written as
	\begin{align*}
		\frac{1}{2}\|A\delta_E\|_2^2 =\frac{1}{2}\delta_E^\top A^\top A\delta_E.
	\end{align*}
	Thus, the estimator depends on \(A\) only through the positive semidefinite matrix \(A^\top A\).
	
	\begin{proposition}[Closed form representation]
		\label{prop:generalized-closed-form}
		The coordinates of the generalized correction indexed by \(E\) satisfy
		\begin{align*}
			(\Sigma_{n,E}+A^\top A)\hat{\delta}_{A,E}
			=
			\lambda_{\mathrm L}s.
		\end{align*}
		Consequently,
		\begin{align*}
			\hat{\delta}_{A,E}
			&=
			(\Sigma_{n,E}+A^\top A)^{-1}\lambda_{\mathrm L}s,
			\qquad
			\hat{\delta}_{A,-E}=0,
		\end{align*}
		and
		\begin{align*}
			\hat{\beta}_{A,E}
			&=
			\hat{\beta}_{\mathrm L,E}
			+
			(\Sigma_{n,E}+A^\top A)^{-1}\lambda_{\mathrm L}s,
			\qquad
			\hat{\beta}_{A,-E}=0.
		\end{align*}
	\end{proposition}
	
	\begin{proof}
		Because \(\delta_{-E}=0\), the objective defining
		\(\hat{\delta}_A\) reduces to
		\begin{align*}
			\frac{1}{2n}
			\|y-X\hat{\beta}_{\mathrm L}-X_E\delta_E\|_2^2
			+
			\frac{1}{2}\|A\delta_E\|_2^2.
		\end{align*}
		Differentiating with respect to \(\delta_E\) gives
		\begin{align*}
			-\frac{1}{n}X_E^\top
			(
			y-X\hat{\beta}_{\mathrm L}-X_E\delta_E
			)
			+
			A^\top A\delta_E.
		\end{align*}
		Setting the gradient equal to zero at
		\(\hat{\delta}_{A,E}\) gives
		\begin{align*}
			(
			\frac{1}{n}X_E^\top X_E+A^\top A
			)
			\hat{\delta}_{A,E}
			=
			\frac{1}{n}X_E^\top
			(y-X\hat{\beta}_{\mathrm L}).
		\end{align*}
		Using
		\begin{align*}
			\Sigma_{n,E}
			=
			\frac{1}{n}X_E^\top X_E
		\end{align*}
		and the Lasso KKT identity
		\begin{align*}
			\frac{1}{n}X_E^\top
			(y-X\hat{\beta}_{\mathrm L})
			=
			\lambda_{\mathrm L}s,
		\end{align*}
		we obtain
		\begin{align*}
			(\Sigma_{n,E}+A^\top A)\hat{\delta}_{A,E}
			=
			\lambda_{\mathrm L}s.
		\end{align*}
		By Assumption~\ref{ass:general-position},
		\(\Sigma_{n,E}\) is positive definite whenever \(|E|>0\),
		while \(A^\top A\) is positive semidefinite. Hence,
		\(\Sigma_{n,E}+A^\top A\) is positive definite and therefore
		invertible. It follows that
		\begin{align*}
			\hat{\delta}_{A,E}
			=
			(\Sigma_{n,E}+A^\top A)^{-1}\lambda_{\mathrm L}s,
			\qquad
			\hat{\delta}_{A,-E}=0.
		\end{align*}
		Finally, since
		\(\hat{\beta}_A=\hat{\beta}_{\mathrm L}+\hat{\delta}_A\),
		we have
		\begin{align*}
			\hat{\beta}_{A,E}
			&=
			\hat{\beta}_{\mathrm L,E}
			+
			(\Sigma_{n,E}+A^\top A)^{-1}\lambda_{\mathrm L}s,
			\qquad
			\hat{\beta}_{A,-E}=0.
		\end{align*}
		The case \(E=\varnothing\) is covered by the zero dimensional
		conventions introduced above.
	\end{proof}
	
	\section{Basic Properties and Special Cases}
	\label{sec:basic-properties}
	
	\begin{proposition}[Existence and uniqueness]
		\label{prop:generalized-existence-uniqueness}
		Under Assumption~\ref{ass:general-position}, the generalized Lasso--Ridge estimator is uniquely defined for every
		\(A\in\mathbb{R}^{q\times |E|}\).
	\end{proposition}
	
	\begin{proof}
		If \(E=\varnothing\), the constraint \(\delta_{-E}=0\) forces \(\delta=0\), so uniqueness is immediate. If \(E\neq\varnothing\), Assumption~\ref{ass:general-position} implies that \(X_E\) has full column rank, so \(\Sigma_{n,E}\) is positive definite. Since \(A^{\top}A\) is positive semidefinite, \(\Sigma_{n,E}+A^{\top}A\) is positive definite and therefore invertible. The closed form correction in Proposition~\ref{prop:generalized-closed-form} is consequently unique.
	\end{proof}
	
	\begin{remark}
		If \(A_1^{\top}A_1=A_2^{\top}A_2\), then \(A_1\) and \(A_2\) produce the same generalized Lasso--Ridge estimator. Hence, the effective quadratic penalty matrix is \(A^{\top}A\). We next consider several concrete choices of \(A\), which recover some existing procedures as special cases and also suggest new variants within the generalized framework.
	\end{remark}
	
	\subsection{The Original Lasso--Ridge Estimator}
	
	Let \(A=\sqrt{\lambda_{\mathrm R}}\,I_{|E|}\) with
	\(\lambda_{\mathrm R}\geq 0\). Then
	\(\|A\delta_E\|_2^2/2
	=
	\lambda_{\mathrm R}\|\delta_E\|_2^2/2\),
	so Definition~\ref{def:generalized-lasso-ridge} reduces to the
	Lasso--Ridge estimator. According to the closed form representation in Proposition~\ref{prop:generalized-closed-form},
	the corresponding correction is
	\begin{align*}
		\hat{\delta}_{A,E}
		=
		(
		\Sigma_{n,E}
		+
		\lambda_{\mathrm R}I_{|E|}
		)^{-1}
		\lambda_{\mathrm L}s.
	\end{align*}
	
	\subsection{Least Squares Refitting on the Lasso Equicorrelation Set}
	
	If \(A=0\), the penalty applied in the second stage vanishes. In this case,
	\begin{align*}
		\hat{\beta}_{\mathrm{LS}} \in\underset{\substack{\beta\in\mathbb{R}^p \\
				\beta_{-E}=0}}{\arg\min}
		\frac{1}{2n}\|y-X\beta\|_2^2,
	\end{align*}
	so \(\hat{\beta}_{\mathrm{LS}}\) is the least squares refitted estimator on the Lasso equicorrelation set, in the spirit of Lasso based least squares refitting; see, for example, \citet{belloni2013least}. Its active coordinates satisfy
	\begin{align*}
		\hat{\beta}_{\mathrm{LS},E}
		&= (X_E^{\top}X_E)^{-1}X_E^{\top}y \\
		&= \Sigma_{n,E}^{-1}
		\left(\frac{1}{n}X_E^\top y\right).
	\end{align*}
	Using the Lasso KKT identity from Proposition~\ref{prop:lasso-kkt},
	\begin{align*}
		\frac{1}{n}X_E^\top
		(y-X\hat{\beta}_{\mathrm L})
		=
		\lambda_{\mathrm L}s,
	\end{align*}
	we obtain
	\begin{align*}
		\hat{\beta}_{\mathrm{LS},E}
		=
		\hat{\beta}_{\mathrm L,E}
		+
		\Sigma_{n,E}^{-1}\lambda_{\mathrm L}s.
	\end{align*}
	This coincides with the closed form representation in
	Proposition~\ref{prop:generalized-closed-form} when \(A=0\).
	
	\subsection{Gram Proportional Penalties}
	
	A useful subclass is obtained by taking the quadratic penalty matrix proportional to the restricted Gram matrix. Let \(\tau\geq0\) and suppose that \(A^\top A=\tau\Sigma_{n,E}\). Proposition~\ref{prop:generalized-closed-form} then gives
	\begin{align*}
		\hat{\delta}_{A,E} &= ( \Sigma_{n,E} + \tau\Sigma_{n,E} )^{-1} \lambda_{\mathrm L}s \\
		&= \frac{1}{1+\tau} \Sigma_{n,E}^{-1} \lambda_{\mathrm L}s.
	\end{align*}
	Consequently,
	\begin{align*}
		\hat{\beta}_{A,E} = \hat{\beta}_{\mathrm L,E} + \frac{1}{1+\tau} \Sigma_{n,E}^{-1} \lambda_{\mathrm L}s.
	\end{align*}
	Equivalently, define \(\phi\coloneqq\tau/(1+\tau)\in[0,1)\), so that \(\tau=\phi/(1-\phi)\). Hence,
	\begin{align*}
		A^\top A &= \frac{\phi}{1-\phi} \Sigma_{n,E}.
	\end{align*}
	For this choice,
	\begin{align*}
		\hat{\delta}_{A,E} &= (1-\phi) \Sigma_{n,E}^{-1} \lambda_{\mathrm L}s, \qquad \hat{\delta}_{A,-E}=0,
	\end{align*}
	and therefore
	\begin{align*}
		\hat{\beta}_{A,E} = \hat{\beta}_{\mathrm L,E} + (1-\phi) \Sigma_{n,E}^{-1} \lambda_{\mathrm L}s.
	\end{align*}
	Since
	\begin{align*}
		\hat{\beta}_{\mathrm L,E} = (X_E^\top X_E)^{-1} ( X_E^\top y - n\lambda_{\mathrm L}s ),
	\end{align*}
	it follows that
	\begin{align*}
		\hat{\beta}_{A,E} &= (X_E^\top X_E)^{-1} ( X_E^\top y - n\lambda_{\mathrm L}s ) + n(1-\phi) (X_E^\top X_E)^{-1} \lambda_{\mathrm L}s \\
		&= (X_E^\top X_E)^{-1} ( X_E^\top y - n\phi\lambda_{\mathrm L}s ).
	\end{align*}
	Moreover, since
	\begin{align*}
		\hat{\beta}_{\mathrm{LS},E} = (X_E^\top X_E)^{-1} X_E^\top y,
	\end{align*}
	the same estimator admits the interpolation representation
	\begin{align*}
		\hat{\beta}_{A,E} = \phi\hat{\beta}_{\mathrm L,E} + (1-\phi)\hat{\beta}_{\mathrm{LS},E}.
	\end{align*}
	
	\subsection{Relation to the Relaxed Lasso}
	
	Following the relaxation idea of \citet{meinshausen2007relaxed}, let \(\phi\in[0,1]\) denote the relaxation parameter. We consider an equicorrelation-set version of the relaxed Lasso: after the first-stage Lasso determines the set \(E\), define
	\begin{align*}
		\hat{\beta}_{\mathrm{rel}}(\phi) \in\underset{ \substack{ \beta\in\mathbb{R}^p \\
				\beta_{-E}=0 } }{\arg\min} \left\{\frac{1}{2n}\|y-X\beta\|_2^2 + \phi\lambda_{\mathrm L} \|\beta_E\|_1\right\}.
	\end{align*}
	For \(\phi=0\), this is understood as an unpenalized refit by least squares on \(E\), whereas \(\phi=1\) recovers the original Lasso estimator. This differs slightly from the original relaxed Lasso of \citet{meinshausen2007relaxed}, which restricts the second-stage problem to the nonzero support selected by the first-stage Lasso; the two definitions coincide whenever \(\operatorname{supp}(\hat{\beta}_{\mathrm L})=E\).
	
	Now fix \(\phi\in[0,1)\), and suppose that the relaxed Lasso retains the support and sign vector obtained in the first stage, namely
	\begin{align*}
		\operatorname{supp}( \hat{\beta}_{\mathrm{rel}}(\phi) ) &= E, \\
		\operatorname{sign}( \hat{\beta}_{\mathrm{rel},E}(\phi) ) &= s.
	\end{align*}
	Under this condition, the KKT equations on \(E\) are
	\begin{align*}
		\frac{1}{n} X_E^\top ( y-X_E\hat{\beta}_{\mathrm{rel},E}(\phi) ) = \phi\lambda_{\mathrm L}s.
	\end{align*}
	Equivalently,
	\begin{align*}
		\Sigma_{n,E} \hat{\beta}_{\mathrm{rel},E}(\phi) = \frac{1}{n}X_E^\top y - \phi\lambda_{\mathrm L}s.
	\end{align*}
	Since \(\Sigma_{n,E}\) is invertible under Assumption~\ref{ass:general-position},
	\begin{align*}
		\hat{\beta}_{\mathrm{rel},E}(\phi) &= \Sigma_{n,E}^{-1} \left(\frac{1}{n}X_E^\top y - \phi\lambda_{\mathrm L}s\right) \\
		&= (X_E^\top X_E)^{-1} ( X_E^\top y - n\phi\lambda_{\mathrm L}s ).
	\end{align*}
	Using the identity for the Lasso estimator obtained in the first stage from Proposition~\ref{prop:lasso-closed-form},
	\begin{align*}
		\hat{\beta}_{\mathrm L,E} = \Sigma_{n,E}^{-1} \left(\frac{1}{n}X_E^\top y - \lambda_{\mathrm L}s\right),
	\end{align*}
	the relaxed Lasso solution can also be written as
	\begin{align*}
		\hat{\beta}_{\mathrm{rel},E}(\phi) = \hat{\beta}_{\mathrm L,E} + (1-\phi) \Sigma_{n,E}^{-1} \lambda_{\mathrm L}s.
	\end{align*}
	
	This representation coincides with a Gram proportional generalized Lasso--Ridge estimator. Specifically, for \(\phi\in[0,1)\), define
	\begin{align*}
		A_\phi^\top A_\phi = \frac{\phi}{1-\phi} \Sigma_{n,E}.
	\end{align*}
	For example, one may take \(q=|E|\) and
	\begin{align*}
		A_\phi = \sqrt{\frac{\phi}{1-\phi}}\, \Sigma_{n,E}^{1/2},
	\end{align*}
	where \(\Sigma_{n,E}^{1/2}\) denotes the symmetric positive definite square root of \(\Sigma_{n,E}\). Equivalently,
	\begin{align*}
		A_\phi^\top A_\phi = \frac{\phi}{n(1-\phi)} X_E^\top X_E.
	\end{align*}
	By Proposition~\ref{prop:generalized-closed-form},
	\begin{align*}
		\hat{\delta}_{A_\phi,E} &= ( \Sigma_{n,E} + A_\phi^\top A_\phi )^{-1} \lambda_{\mathrm L}s \\
		&= \left(\Sigma_{n,E} + \frac{\phi}{1-\phi} \Sigma_{n,E}\right)^{-1} \lambda_{\mathrm L}s \\
		&= (1-\phi) \Sigma_{n,E}^{-1} \lambda_{\mathrm L}s.
	\end{align*}
	Therefore,
	\begin{align*}
		\hat{\beta}_{A_\phi,E} &= \hat{\beta}_{\mathrm L,E} + (1-\phi) \Sigma_{n,E}^{-1} \lambda_{\mathrm L}s \\
		&= (X_E^\top X_E)^{-1} ( X_E^\top y - n\phi\lambda_{\mathrm L}s ).
	\end{align*}
	Consequently, whenever the relaxed Lasso retains the support \(E\) and sign vector \(s\) obtained in the first stage,
	\begin{align*}
		\hat{\beta}_{A_\phi} = \hat{\beta}_{\mathrm{rel}}(\phi), \qquad \phi\in[0,1).
	\end{align*}
	Thus, on any portion of the relaxed Lasso path for which the support and signs remain unchanged, the relaxed Lasso is represented exactly by a generalized Lasso--Ridge estimator whose quadratic penalty matrix is proportional to the restricted Gram matrix.
	
	The two endpoints have natural interpretations. At \(\phi=0\), we have \(A_0^\top A_0=0\), and both procedures reduce to refitting by least squares on the Lasso equicorrelation set \(E\). As \(\phi\uparrow1\), the factor \(\phi/(1-\phi)\) diverges while \(\hat{\delta}_{A_\phi,E}=(1-\phi)\Sigma_{n,E}^{-1}\lambda_{\mathrm L}s\to0\), so \(\hat{\beta}_{A_\phi,E}\to\hat{\beta}_{\mathrm L,E}\). Thus, \(\phi=1\), corresponding to the original Lasso, is obtained as the limit of the Gram proportional family as the penalty tends to infinity.
	
	\begin{remark}
		The preceding equivalence is conditional on stability of the support and sign pattern for the relaxed Lasso solution obtained in the second stage. In general, the relaxed Lasso may set one or more coordinates in \(E\) to zero, in which case its KKT subgradient need not equal the sign vector \(s\) obtained in the first stage. In contrast, the generalized Lasso--Ridge estimator uses the fixed vector \(s\) obtained in the first stage in its closed form correction. Therefore, the two estimators need not coincide once the support or sign pattern of the relaxed Lasso changes. The equivalence is exact on portions of the relaxed Lasso path for which
		\begin{align*}
			\operatorname{supp}(\hat{\beta}_{\mathrm{rel}}(\phi)) &= E,
			\qquad
			\operatorname{sign}(\hat{\beta}_{\mathrm{rel},E}(\phi)) = s.
		\end{align*}
	\end{remark}
	
	\subsection{Diagonal Penalties}
	
	If \(A=\operatorname{diag}(a_j:j\in E)\), then
	\begin{align*}
		\frac{1}{2}\|A\delta_E\|_2^2 = \frac{1}{2}\sum_{j\in E} a_j^2\delta_j^2.
	\end{align*}
	This form allows the correction vector to be penalized differently across coordinates.
	
	\section{Prediction Improvement: The Positive Gain Term}
	\label{sec:prediction-improvement}
	
	Recall the realized prediction improvement defined in the introduction. Following the same argument used for the Lasso--Ridge estimator in \citet{liu2026prediction}, we obtain the corresponding decomposition for the generalized correction. A full projection-based derivation is provided in Appendix~\ref{app:prediction-decomposition}. Using the linear model, the support restriction \(\hat{\delta}_{A,-E}=0\), and the Lasso KKT identity, we obtain
	\begin{align}
		\Delta(\hat{\beta}_{\mathrm L},\hat{\beta}_A)
		&=2\lambda_{\mathrm L}\langle \hat{\delta}_{A,E},s\rangle
		-\frac{1}{n}\langle X_E\hat{\delta}_{A,E},X_E\hat{\delta}_{A,E}\rangle
		-\frac{2}{n}\langle X_E\hat{\delta}_{A,E},\varepsilon\rangle.
		\label{eq:prediction-decomposition}
	\end{align}
	Under the zero dimensional conventions adopted above, this identity also holds when \(E=\varnothing\). We first control the first two terms. Define
	\begin{align*}
		H_{E,s}(A)
		\coloneqq
		2\lambda_{\mathrm L}\langle \hat{\delta}_{A,E},s\rangle
		-\frac{1}{n}\langle X_E\hat{\delta}_{A,E},\,X_E\hat{\delta}_{A,E}\rangle.
	\end{align*}
	Then
	\begin{align*}
		\Delta(\hat{\beta}_{\mathrm L},\hat{\beta}_A)
		=
		H_{E,s}(A)
		-\frac{2}{n}\langle X_E\hat{\delta}_{A,E},\varepsilon\rangle.
	\end{align*}
	
	\begin{proposition}[Positivity and lower bounds for \(H_{E,s}(A)\)]
		\label{prop:sufficient-penalty-conditions}
		Suppose that Assumption~\ref{ass:general-position} holds and \(E\neq\varnothing\). If
		\begin{align}
			A^\top A\succeq0,
			\label{eq:general-psd-condition}
		\end{align}
		then \(H_{E,s}(A)>0\). If, in addition, there exists \(\tau\geq0\) such that
		\begin{align}
			0\preceq A^\top A\preceq\tau\Sigma_{n,E},
			\label{eq:relative-penalty-condition}
		\end{align}
		then
		\begin{align}
			H_{E,s}(A)
			\geq
			\lambda_{\mathrm L}^2|E|
			\frac{1+2\tau}{(1+\tau)^2}
			\frac{1}{\|\Sigma_{n,E}\|_2}
			>0.
			\label{eq:general-positive-lower-bound}
		\end{align}
	\end{proposition}
	
	\begin{proof}
		By the closed form representation in
		Proposition~\ref{prop:generalized-closed-form},
		\begin{align*}
			(\Sigma_{n,E}+A^\top A)\hat{\delta}_{A,E} =\lambda_{\mathrm L}s.
		\end{align*}
		Multiplying from the left by \(\hat{\delta}_{A,E}^\top\) yields
		\begin{align*}
			\lambda_{\mathrm L}\langle\hat{\delta}_{A,E},s\rangle =\hat{\delta}_{A,E}^\top (\Sigma_{n,E}+A^\top A)\hat{\delta}_{A,E}.
		\end{align*}
		Moreover,
		\begin{align*}
			\frac{1}{n}\langle X_E\hat{\delta}_{A,E},X_E\hat{\delta}_{A,E}\rangle =\hat{\delta}_{A,E}^\top\Sigma_{n,E}\hat{\delta}_{A,E}.
		\end{align*}
		Therefore, substituting these identities into the definition of \(H_{E,s}(A)\) gives
		\begin{align}
			H_{E,s}(A) =\hat{\delta}_{A,E}^\top (\Sigma_{n,E}+2A^\top A)\hat{\delta}_{A,E}. \label{eq:H-general-matrix}
		\end{align}
		Under Assumption~\ref{ass:general-position}, \(\Sigma_{n,E}\succ0\). Hence, if \(A^\top A\succeq0\), then
		\begin{align*}
			\Sigma_{n,E}+2A^\top A\succ0.
		\end{align*}
		Since \(\lambda_{\mathrm L}>0\), \(s\neq0\), and \(\Sigma_{n,E}+A^\top A\) is invertible,
		\begin{align*}
			\hat{\delta}_{A,E} =(\Sigma_{n,E}+A^\top A)^{-1}\lambda_{\mathrm L}s \neq0.
		\end{align*}
		It follows from~\eqref{eq:H-general-matrix} that
		\begin{align*}
			H_{E,s}(A)>0.
		\end{align*}
		
		To derive the quantitative lower bound, define
		\begin{align*}
			R_E\coloneqq \Sigma_{n,E}^{-1/2}A^\top A\Sigma_{n,E}^{-1/2}.
		\end{align*}
		Condition~\eqref{eq:relative-penalty-condition} implies, by congruence with
		\(\Sigma_{n,E}^{-1/2}\),
		\begin{align*}
			0 \preceq R_E
			= \Sigma_{n,E}^{-1/2}A^\top A\Sigma_{n,E}^{-1/2}
			\preceq \tau I_{|E|}.
		\end{align*}
		Furthermore,
		\begin{align*}
			\Sigma_{n,E}+A^\top A
			&=\Sigma_{n,E}^{1/2}(I_{|E|}+R_E)\Sigma_{n,E}^{1/2}, \\
			\Sigma_{n,E}+2A^\top A
			&=\Sigma_{n,E}^{1/2}(I_{|E|}+2R_E)\Sigma_{n,E}^{1/2}.
		\end{align*}
		First, substituting the closed form representation from
		Proposition~\ref{prop:generalized-closed-form}
		into~\eqref{eq:H-general-matrix} gives
		\begin{align*}
			H_{E,s}(A)
			&=
			\lambda_{\mathrm L}^2
			s^\top
			(\Sigma_{n,E}+A^\top A)^{-1}
			(\Sigma_{n,E}+2A^\top A)
			(\Sigma_{n,E}+A^\top A)^{-1}
			s.
		\end{align*}
		Using the two identities above, we obtain
		\begin{align*}
			H_{E,s}(A)
			&=
			\lambda_{\mathrm L}^2
			s^\top
			\Sigma_{n,E}^{-1/2}
			(I_{|E|}+R_E)^{-1}
			(I_{|E|}+2R_E)
			(I_{|E|}+R_E)^{-1}
			\Sigma_{n,E}^{-1/2}s.
		\end{align*}
		By the spectral theorem and the bound \(0\preceq R_E\preceq \tau I_{|E|}\),
		\begin{align*}
			R_E
			=
			Q\operatorname{diag}(\nu_1,\ldots,\nu_{|E|})Q^\top,
			\qquad
			0\leq\nu_i\leq\tau.
		\end{align*}
		Then
		\begin{align*}
			& (I_{|E|}+R_E)^{-1}(I_{|E|}+2R_E)(I_{|E|}+R_E)^{-1} \\
			& \qquad= Q\operatorname{diag}\left(\frac{1+2\nu_1}{(1+\nu_1)^2}, \ldots, \frac{1+2\nu_{|E|}}{(1+\nu_{|E|})^2}\right)Q^\top.
		\end{align*}
		The function \(f(x)\coloneqq(1+2x)/(1+x)^2\) is decreasing on \([0,\infty)\), since
		\begin{align*}
			f'(x)=-\frac{2x}{(1+x)^3}\leq0.
		\end{align*}
		Hence,
		\begin{align*}
			\frac{1+2\nu_i}{(1+\nu_i)^2} \geq \frac{1+2\tau}{(1+\tau)^2}, \qquad i=1,\ldots,|E|.
		\end{align*}
		Therefore,
		\begin{align*}
			H_{E,s}(A) \geq \lambda_{\mathrm L}^2 \frac{1+2\tau}{(1+\tau)^2} s^\top\Sigma_{n,E}^{-1}s.
		\end{align*}
		Since
		\begin{align*}
			\Sigma_{n,E}^{-1} \succeq \frac{1}{\|\Sigma_{n,E}\|_2}I_{|E|}
		\end{align*}
		and \(\| s\|_2^2=|E|\),
		\begin{align*}
			s^\top\Sigma_{n,E}^{-1}s \geq \frac{|E|}{\|\Sigma_{n,E}\|_2}.
		\end{align*}
		Consequently,
		\begin{align*}
			H_{E,s}(A) \geq \lambda_{\mathrm L}^2|E| \frac{1+2\tau}{(1+\tau)^2} \frac{1}{\|\Sigma_{n,E}\|_2} >0.
		\end{align*}
	\end{proof}
	
	\begin{remark}
		Since \(A^\top A\succeq0\) for every real matrix \(A\), condition~\eqref{eq:general-psd-condition} is automatically satisfied. Thus, under Assumption~\ref{ass:general-position} and \(E\neq\varnothing\), every generalized Lasso--Ridge estimator with quadratic penalty
		\begin{align*}
			\frac{1}{2}\|A\delta_E\|_2^2
		\end{align*}
		satisfies
		\begin{align*}
			H_{E,s}(A)>0.
		\end{align*}
		Condition~\eqref{eq:relative-penalty-condition} is stronger and provides an explicit positive lower bound.
	\end{remark}
	
	\subsection{Some Concrete Examples}
	
	Throughout this subsection, suppose \(E\neq\varnothing\).
	
	\paragraph{Least squares refitting.}
	For least squares refitting, \(A^\top A=0\). Thus, condition~\eqref{eq:relative-penalty-condition} holds with \(\tau=0\), and
	\begin{align*}
		H_{E,s}(0) \geq \frac{\lambda_{\mathrm L}^2|E|}{\|\Sigma_{n,E}\|_2} >0.
	\end{align*}
	
	\paragraph{Lasso--Ridge penalty.}
	For the Lasso--Ridge estimator, \(A^\top A=\lambda_{\mathrm R}I_{|E|}\) with \(\lambda_{\mathrm R}\geq0\). Since \(A^\top A\succeq0\), we have \(H_{E,s}(A)>0\). Let \(\mu_{\min,E}\coloneqq\lambda_{\min}(\Sigma_{n,E})>0\). Since \(\Sigma_{n,E}\succeq\mu_{\min,E}I_{|E|}\), we have
	\begin{align*}
		A^\top A =\lambda_{\mathrm R}I_{|E|} \preceq \frac{\lambda_{\mathrm R}}{\mu_{\min,E}}\Sigma_{n,E}.
	\end{align*}
	Thus, condition~\eqref{eq:relative-penalty-condition} holds with \(\tau_{\mathrm R}\coloneqq\lambda_{\mathrm R}/\mu_{\min,E}\), and
	\begin{align*}
		H_{E,s}(A) \geq \lambda_{\mathrm L}^2|E| \frac{1+2\tau_{\mathrm R}}{(1+\tau_{\mathrm R})^2} \frac{1}{\|\Sigma_{n,E}\|_2} >0.
	\end{align*}
	
	\paragraph{Gram proportional penalty.}
	Suppose that \(A^\top A=\tau\Sigma_{n,E}\) with \(\tau\geq0\). Then \(0\preceq A^\top A\preceq\tau\Sigma_{n,E}\), so Proposition~\ref{prop:sufficient-penalty-conditions} gives
	\begin{align*}
		H_{E,s}(A) \geq \lambda_{\mathrm L}^2|E| \frac{1+2\tau}{(1+\tau)^2} \frac{1}{\|\Sigma_{n,E}\|_2} >0.
	\end{align*}
	In this case,
	\begin{align*}
		H_{E,s}(A) = \lambda_{\mathrm L}^2 \frac{1+2\tau}{(1+\tau)^2} s^\top\Sigma_{n,E}^{-1}s.
	\end{align*}
	
	\paragraph{Diagonal penalty.}
	Suppose that \(A^\top A=D_E\coloneqq\operatorname{diag}(d_j:j\in E)\) with \(d_j\geq0\). Then \(A^\top A=D_E\succeq0\), so \(H_{E,s}(A)>0\). Define \(d_{\max}\coloneqq\max_{j\in E}d_j\). Since
	\begin{align*}
		A^\top A=D_E \preceq d_{\max}I_{|E|} \preceq \frac{d_{\max}}{\mu_{\min,E}}\Sigma_{n,E},
	\end{align*}
	condition~\eqref{eq:relative-penalty-condition} holds with \(\tau_D\coloneqq d_{\max}/\mu_{\min,E}\). Therefore,
	\begin{align*}
		H_{E,s}(A) \geq \lambda_{\mathrm L}^2|E| \frac{1+2\tau_D}{(1+\tau_D)^2} \frac{1}{\|\Sigma_{n,E}\|_2} >0.
	\end{align*}
	
	\paragraph{Combined ridge and Gram penalty.}
	Suppose that \(A^\top A=\lambda_{\mathrm R}I_{|E|}+\tau\Sigma_{n,E}\) with \(\lambda_{\mathrm R}\geq0\) and \(\tau\geq0\). Then \(A^\top A\succeq0\), so \(H_{E,s}(A)>0\). Moreover,
	\begin{align*}
		A^\top A \preceq \left(\tau+\frac{\lambda_{\mathrm R}}{\mu_{\min,E}}\right)\Sigma_{n,E}.
	\end{align*}
	Defining \(\tau_*\coloneqq\tau+\lambda_{\mathrm R}/\mu_{\min,E}\), we obtain
	\begin{align*}
		H_{E,s}(A) \geq \lambda_{\mathrm L}^2|E| \frac{1+2\tau_*}{(1+\tau_*)^2} \frac{1}{\|\Sigma_{n,E}\|_2} >0.
	\end{align*}
	
	\paragraph{Low rank penalty.}
	Suppose that \(A^\top A=UU^\top\) for some matrix \(U\). Since \(UU^\top\succeq0\), we have \(H_{E,s}(A)>0\). Define
	\begin{align*}
		\tau_U\coloneqq
		\|\Sigma_{n,E}^{-1/2}UU^\top\Sigma_{n,E}^{-1/2}\|_2.
	\end{align*}
	By the definition of \(\tau_U\),
	\begin{align*}
		\Sigma_{n,E}^{-1/2}UU^\top\Sigma_{n,E}^{-1/2}
		\preceq \tau_U I,
	\end{align*}
	and hence
	\begin{align*}
		A^\top A=UU^\top\preceq\tau_U\Sigma_{n,E}.
	\end{align*}
	Therefore,
	\begin{align*}
		H_{E,s}(A)
		\geq \lambda_{\mathrm L}^2|E|
		\frac{1+2\tau_U}{(1+\tau_U)^2}
		\frac{1}{\|\Sigma_{n,E}\|_2}
		>0.
	\end{align*}
	
	\section{Localization Control of the Noise Term}
	\label{sec:refined-localization-lasso-correction}
	
	The previous section shows that the generalized correction creates a positive gain contribution to prediction improvement. To obtain an expected prediction comparison, it remains to control
	\begin{align*}
		\mathbb E\left[
		\frac{2}{n}\langle X_E\hat{\delta}_{A,E},\varepsilon\rangle
		\right],
	\end{align*}
	where both the equicorrelation set \(E\) and its sign vector \(s\) are data dependent. In this section, \(X\) and \(\lambda_{\mathrm L}\) remain deterministic as in Section~\ref{sec:setup}, and all probabilities and expectations are taken with respect to \(\varepsilon\). Under Assumption~\ref{ass:general-position}, every realized equicorrelation set satisfies
	\begin{align*}
		\operatorname{rank}(X_E) &=|E|, \\
		|E| &\leq m_n,\qquad m_n\coloneqq\min\{n,p\}.
	\end{align*}
	Recall that
	\begin{align*}
		\hat{\delta}_{A,E}
		&=(\Sigma_{n,E}+A^\top A)^{-1}\lambda_{\mathrm L}s, \\
		\Sigma_{n,E} &\coloneqq\frac1nX_E^\top X_E.
	\end{align*}
	For the localization argument we specialize the general penalty construction to a deterministic family of penalty matrices indexed by the candidate support. For each candidate support \(F\), the matrix \(A_F\) is fixed, while the realized matrix \(A=A_E\) may still be random through the random support \(E\). This restriction is what permits the selection induced noise term to be controlled by a finite deterministic family of linear forms.
	
	\subsection{Attainable signed equicorrelation models}
	For fixed \(X\) and \(\lambda_{\mathrm L}>0\), define the deterministic class of attainable signed equicorrelation models
	\begin{align*}
		\mathcal C_n\coloneqq\Bigl\{(F,z): & \text{ there exists }y\in\mathbb R^n\text{ for which }F\text{ is the Lasso equicorrelation set} \\
		& \text{and }z\text{ is its equicorrelation sign vector}\Bigr\}.
	\end{align*}
	By construction, the realized signed equicorrelation model satisfies
	\begin{align*}
		(E,s)\in\mathcal C_n.
	\end{align*}
	The role of \(\mathcal C_n\) is geometric rather than probabilistic: the class may still be large, but it restricts the argument to signed supports that can actually arise as Lasso equicorrelation models. The effective complexity reduction below instead comes from localization around the reference support \(S_0\). General position is used only for attainable equicorrelation supports. In particular, for every \((F,z)\in\mathcal C_n\),
	\begin{align*}
		\operatorname{rank}(X_F) &=|F|, \\
		|F| &\leq m_n.
	\end{align*}
	For any \(F\subseteq[p]\), define
	\begin{align*}
		\Sigma_{n,F}\coloneqq\frac{1}{n}X_F^\top X_F.
	\end{align*}
	For every nonempty attainable support \(F\), general position gives
	\begin{align*}
		\Sigma_{n,F}\succ0.
	\end{align*}
	Let
	\begin{align*}
		\mathfrak F_n\coloneqq\{F\subseteq[p]:(F,z)\in\mathcal C_n\text{ for some }z\}
	\end{align*}
	denote the deterministic class of attainable equicorrelation supports.
	Let \(S_0\) denote the deterministic reference support, which in the applications below is the true support \(\operatorname{supp}(\beta_0)\), and write
	\begin{align*}
		s_0\coloneqq|S_0|.
	\end{align*}
	Let \(z_0\in\{-1,1\}^{s_0}\) be a fixed deterministic reference sign vector, for example \(z_0=\operatorname{sign}(\beta_{0,S_0})\) when \(S_0=\operatorname{supp}(\beta_0)\). Suppose that a deterministic family of penalty matrices
	\begin{align*}
		\{A_F:F\in\mathfrak F_n\cup\{S_0\}\},
		\qquad A_F\in\mathbb R^{q_F\times|F|},
	\end{align*}
	with \(q_F\geq0\), is specified, and write \(A\equiv A_E\) for the correction matrix corresponding to the realized support. Thus, \(A_F\) is fixed for each candidate support \(F\), while \(A_E\) is random only through the random support \(E\). This includes choices such as \(A_F=X_F\). For every nonempty attainable signed model \((F,z)\in\mathcal C_n\), define
	\begin{align*}
		a_{F,z}\coloneqq\frac{2\lambda_{\mathrm L}}{n}X_F(\Sigma_{n,F}+A_F^\top A_F)^{-1}z,
	\end{align*}
	and set
	\begin{align*}
		a_{\varnothing,\varnothing}\coloneqq0.
	\end{align*}
	For every nonempty attainable \(F\),
	\begin{align*}
		\Sigma_{n,F}+A_F^\top A_F\succ0,
	\end{align*}
	because \(\Sigma_{n,F}\succ0\) and \(A_F^\top A_F\succeq0\). Hence the inverse above exists automatically for every attainable nonempty support.
	Assume \(\Sigma_{n,S_0}+A_{S_0}^\top A_{S_0}\) is nonsingular and, when \(S_0\neq\varnothing\), define
	\begin{align*}
		a_0\coloneqq\frac{2\lambda_{\mathrm L}}{n}X_{S_0}(\Sigma_{n,S_0}+A_{S_0}^\top A_{S_0})^{-1}z_0,
	\end{align*}
	with \(a_0=0\) when \(S_0=\varnothing\). Finally, since \(A\equiv A_E\), the realized correction satisfies
	\begin{align*}
		\frac{2}{n}\langle X_E\hat{\delta}_{A,E},\varepsilon\rangle=\langle a_{E,s},\varepsilon\rangle.
	\end{align*}
	
	\subsection{Localization and effective correction amplification}
	
	For sets \(A,B\subseteq[p]\), write \(A\triangle B\coloneqq(A\setminus B)\cup(B\setminus A)\) for their symmetric difference. For \(r\geq0\), define the equicorrelation support localization event
	\begin{align*}
		G_n(r) &\coloneqq \{|E\triangle S_0| \leq r\}.
	\end{align*}
	Let
	\begin{align*}
		\mathcal R_n &\subseteq \{0,1,\ldots,p\}
	\end{align*}
	be nonempty and deterministic. Suppose that, for every
	\(r\in\mathcal R_n\),
	\begin{align*}
		\mathbb P(G_n(r)^c) &\leq \delta_n(r), \\
		\delta_n(r) & \in [0,1].
	\end{align*}
	In particular, if \(p\in\mathcal R_n\), then \(G_n(p)\) is the whole sample space, so one may take \(\delta_n(p)=0\).
	Define
	\begin{align*}
		q_n(r) &\coloneqq \mathbb P((E,s)\neq(S_0,z_0), \, G_n(r)),
	\end{align*}
	and
	\begin{align*}
		\psi_p(r) &\coloneqq \begin{cases} 0, & r=0, \\[1mm]
			r\log\!\left(\dfrac{ep}{r}\right), & 1\leq r\leq p.\end{cases}
	\end{align*}
	
	For every nonempty \(F\in\mathfrak F_n\cup\{S_0\}\), define
	\begin{align*}
		\Gamma_F &\coloneqq \|
		(\Sigma_{n,F}+A_F^\top A_F)^{-1}
		\Sigma_{n,F}
		(\Sigma_{n,F}+A_F^\top A_F)^{-1}
		\|_2,
	\end{align*}
	with
	\begin{align*}
		\Gamma_\varnothing &\coloneqq 0.
	\end{align*}
	Define the reference amplification
	\begin{align*}
		\Gamma_0 &\coloneqq \Gamma_{S_0},
	\end{align*}
	the local attainable equicorrelation amplification
	\begin{align*}
		\Gamma_n(r) &\coloneqq \max_{\substack{ (F,z)\in\mathcal C_n \\
				|F\triangle S_0|\leq r }} \Gamma_F,
	\end{align*}
	with the convention that the maximum over an empty class is zero, and the global attainable equicorrelation amplification
	\begin{align*}
		\overline\Gamma_n &\coloneqq \max_{(F,z)\in\mathcal C_n} \Gamma_F.
	\end{align*}
	All these quantities are deterministic. General position guarantees only that
	\(\overline\Gamma_n<\infty\) for each fixed problem instance. It does not
	imply a useful uniform with respect to \(n\) upper bound.
	
	\begin{assumption}[Sub-Gaussian noise]
		\label{ass:subgaussian-noise-lasso-correction}
		The noise vector \(\varepsilon\in\mathbb R^n\) satisfies
		\(\mathbb E\varepsilon=0\) and, for every fixed \(u\in\mathbb R^n\) and
		\(t\in\mathbb R\),
		\begin{align*}
			\mathbb E\exp(tu^\top\varepsilon) &\leq \exp\left(\frac{ \sigma^2t^2\|u\|_2^2 }{2}\right).
		\end{align*}
	\end{assumption}
	
	\subsection{Expectation bound}
	
	\begin{proposition}[Localized expectation bound over attainable equicorrelation models]
		\label{prop:localized-lasso-correction-expectation}
		Suppose Assumption~\ref{ass:subgaussian-noise-lasso-correction} and the equicorrelation support localization
		bounds above hold. For
		\(r\in\mathcal R_n\), define
		\begin{align*}
			k_n(r) &\coloneqq \min\{m_n,s_0+r\}.
		\end{align*}
		Set
		\begin{align*}
			\mathcal L_n(r) \coloneqq{} & \frac{\sigma\lambda_{\mathrm L}}{\sqrt n} (\sqrt{k_n(r)\Gamma_n(r)} + \sqrt{s_0\Gamma_0}) q_n(r) \\
			& \times \sqrt{ k_n(r) + \psi_p(r) + \log\left(\frac{1}{q_n(r)}\right) },
		\end{align*}
		and
		\begin{align*}
			\mathcal N_n(r) \coloneqq{} & \frac{\sigma\lambda_{\mathrm L}}{\sqrt n} (\sqrt{m_n\overline\Gamma_n} + \sqrt{s_0\Gamma_0}) \delta_n(r) \\
			& \times \sqrt{ m_n\log\left(\frac{ep}{m_n}\right) + \log\left(\frac{1}{\delta_n(r)}\right) }.
		\end{align*}
		Terms of the form
		\(u\sqrt{c+\log(1/u)}\) are interpreted as zero when \(u=0\).
		Then
		\begin{align*}
			\left| \mathbb E\left[ \frac{2}{n} \langle X_E\hat{\delta}_{A,E}, \varepsilon \rangle \right] \right| & \lesssim \inf_{r\in\mathcal R_n} \left\{\mathcal L_n(r) + \mathcal N_n(r)\right\}.
		\end{align*}
		The implicit constant is universal and independent of \(r\).
	\end{proposition}
	
	\begin{proof}
		\medskip\noindent\emph{Step 1: Centering and decomposition.}
		Recall that \(a_0\) is the deterministic reference vector associated with \((S_0,z_0)\). We first rewrite the stochastic linear term so that it vanishes whenever the realized signed model coincides with the reference model.
		Because \(a_0\) is deterministic and \(\mathbb E\varepsilon=0\),
		\begin{align*}
			\mathbb E\langle a_0,\varepsilon\rangle
			&=\langle a_0,\mathbb E\varepsilon\rangle
			=0.
		\end{align*}
		Therefore subtracting \(\langle a_0,\varepsilon\rangle\) does not change the expectation that we want to bound. Define
		\begin{align*}
			Y
			&\coloneqq
			\frac{2}{n}\langle X_E\hat{\delta}_{A,E},\varepsilon\rangle
			-\langle a_0,\varepsilon\rangle.
		\end{align*}
		Using
		\begin{align*}
			\frac{2}{n}\langle X_E\hat{\delta}_{A,E},\varepsilon\rangle
			&=\langle a_{E,s},\varepsilon\rangle,
		\end{align*}
		we obtain
		\begin{align*}
			Y
			&=\langle a_{E,s}-a_0,\varepsilon\rangle.
		\end{align*}
		Hence \(Y=0\) whenever \((E,s)=(S_0,z_0)\), and therefore
		\begin{align*}
			Y
			&=Y\mathbf1_{\{(E,s)\neq(S_0,z_0)\}}.
		\end{align*}
		Moreover,
		\begin{align*}
			\mathbb EY
			&=\frac{2}{n}\mathbb E\bigl[\langle X_E\hat{\delta}_{A,E},\varepsilon\rangle\bigr].
		\end{align*}
		
		Fix \(r\in\mathcal R_n\). Then
		\begin{align*}
			|\mathbb EY|
			&\leq
			\left|\mathbb E\left[Y\mathbf1_{G_n(r)}\right]\right|
			+
			\left|\mathbb E\left[Y\mathbf1_{G_n(r)^c}\right]\right|.
		\end{align*}
		\medskip\noindent\emph{Step 2: Localized event.}
		Define the deterministic local candidate class
		\begin{align*}
			\mathcal A_r &\coloneqq \{(F,z)\in\mathcal C_n:\ |F\triangle S_0|\leq r,\ (F,z)\neq(S_0,z_0)\},
		\end{align*}
		and define the event on which the realized signed model is locally admissible but differs from the reference model:
		\begin{align*}
			B_r &\coloneqq \{(E,s)\neq(S_0,z_0)\}\cap G_n(r).
		\end{align*}
		If \(\mathcal A_r=\varnothing\), the localized contribution is zero, so assume \(\mathcal A_r\neq\varnothing\). On \(B_r\), the realized signed model belongs to \(\mathcal A_r\), and
		\begin{align*}
			|Y|\mathbf1_{G_n(r)}
			&\leq \mathbf1_{B_r}\max_{(F,z)\in\mathcal A_r}
			|\langle a_{F,z}-a_0,\varepsilon\rangle|.
		\end{align*}
		For every \((F,z)\in\mathcal A_r\), localization and attainability give
		\begin{align*}
			|F| &\leq \min\{m_n,s_0+r\}=k_n(r).
		\end{align*}
		By Lemma~\ref{lem:correction-vector-norm-bound} in Appendix~\ref{app:technical-lemmas},
		\begin{align*}
			\|a_{F,z}\|_2 &\leq \frac{2\lambda_{\mathrm L}}{\sqrt n}\sqrt{k_n(r)\Gamma_n(r)}, \\
			\|a_0\|_2 &\leq \frac{2\lambda_{\mathrm L}}{\sqrt n}\sqrt{s_0\Gamma_0}.
		\end{align*}
		Hence, by the triangle inequality,
		\begin{align*}
			\|a_{F,z}-a_0\|_2
			&\leq \frac{2\lambda_{\mathrm L}}{\sqrt n}
			(\sqrt{k_n(r)\Gamma_n(r)}+\sqrt{s_0\Gamma_0}).
		\end{align*}
		For a fixed candidate \((F,z)\), the vector \(a_{F,z}-a_0\) is deterministic. Assumption~\ref{ass:subgaussian-noise-lasso-correction} therefore implies that
		\begin{align*}
			Z_{F,z} &\coloneqq \langle a_{F,z}-a_0,\varepsilon\rangle
		\end{align*}
		is sub-Gaussian with tail
		\begin{align*}
			\mathbb P(|Z_{F,z}|>t)
			&\leq 2\exp\left(-\frac{t^2}{2v_{\mathrm{loc}}(r)^2}\right),
		\end{align*}
		where, temporarily within this proof,
		\begin{align*}
			v_{\mathrm{loc}}(r)
			&\coloneqq \frac{2\sigma\lambda_{\mathrm L}}{\sqrt n}
			(\sqrt{k_n(r)\Gamma_n(r)}+\sqrt{s_0\Gamma_0}).
		\end{align*}
		
		We also need the number of candidate linear forms. Since \(F\mapsto F\triangle S_0\) is injective, the number of supports satisfying \(|F\triangle S_0|\leq r\) is at most
		\begin{align*}
			\sum_{j=0}^r\binom pj.
		\end{align*}
		Every such attainable support has at most \(k_n(r)\) coordinates and hence at most \(2^{k_n(r)}\) possible sign vectors. Therefore,
		\begin{align*}
			|\mathcal A_r|
			&\leq 2^{k_n(r)}\sum_{j=0}^r\binom{p}{j}.
		\end{align*}
		For \(r=0\), the bound \(\log(2|\mathcal A_0|)\lesssim k_n(0)+\psi_p(0)\) is immediate. For \(1\leq r\leq p\), Lemma~\ref{lem:binomial-sum-bound} gives
		\begin{align*}
			\sum_{j=0}^r\binom{p}{j}
			&\leq \left(\frac{ep}{r}\right)^r
		\end{align*}
		and hence
		\begin{align*}
			\log(2|\mathcal A_r|)
			&\leq \log 2+k_n(r)\log 2+r\log\left(\frac{ep}{r}\right) \\
			&=\log 2+k_n(r)\log 2+\psi_p(r) \\
			&\lesssim k_n(r)+\psi_p(r).
		\end{align*}
		
		If \(q_n(r)=0\), then \(B_r\) is a null event and the localized contribution is zero, so assume \(q_n(r)>0\).
		Applying Lemma~\ref{lem:rare-event-maximum-lasso-correction} from Appendix~\ref{app:technical-lemmas} with \(B=B_r\), \(\alpha=q_n(r)\), \(N=|\mathcal A_r|\), and \(v=v_{\mathrm{loc}}(r)\) yields
		\begin{align*}
			|\mathbb E[Y\mathbf1_{G_n(r)}]|
			&\leq \mathbb E[\mathbf1_{B_r}\max_{(F,z)\in\mathcal A_r}|Z_{F,z}|] \\
			& \lesssim v_{\mathrm{loc}}(r)\,q_n(r)
			\sqrt{\log\left(\frac{2|\mathcal A_r|}{q_n(r)}\right)}.
		\end{align*}
		Now separate the two contributions inside the logarithm:
		\begin{align*}
			\log\left(\frac{2|\mathcal A_r|}{q_n(r)}\right)
			&= \log(2|\mathcal A_r|)+\log\left(\frac1{q_n(r)}\right) \\
			& \lesssim k_n(r)+\psi_p(r)+\log\left(\frac1{q_n(r)}\right).
		\end{align*}
		Substituting the expression for \(v_{\mathrm{loc}}(r)\) and absorbing universal constants into \(\lesssim\), we obtain exactly
		\begin{align*}
			|\mathbb E[Y\mathbf1_{G_n(r)}]|
			& \lesssim \frac{\sigma\lambda_{\mathrm L}}{\sqrt n}
			(\sqrt{k_n(r)\Gamma_n(r)}+\sqrt{s_0\Gamma_0})q_n(r) \\
			& \qquad\times
			\sqrt{k_n(r)+\psi_p(r)+\log\left(\frac1{q_n(r)}\right)} \\
			&= \mathcal L_n(r).
		\end{align*}
		\medskip\noindent\emph{Step 3: Nonlocalized event.}
		On \(G_n(r)^c\), we no longer know that the selected support lies near \(S_0\). We therefore maximize over the entire attainable class:
		\begin{align*}
			|Y|\mathbf1_{G_n(r)^c}
			&\leq \mathbf1_{G_n(r)^c}\max_{(F,z)\in\mathcal C_n}
			|\langle a_{F,z}-a_0,\varepsilon\rangle|.
		\end{align*}
		For every \((F,z)\in\mathcal C_n\), attainability gives \(|F|\leq m_n\), and the definition of \(\overline\Gamma_n\) gives \(\Gamma_F\leq\overline\Gamma_n\). Hence Lemma~\ref{lem:correction-vector-norm-bound} gives
		\begin{align*}
			\|a_{F,z}\|_2 &\leq \frac{2\lambda_{\mathrm L}}{\sqrt n}\sqrt{m_n\overline\Gamma_n},
		\end{align*}
		and therefore
		\begin{align*}
			\|a_{F,z}-a_0\|_2
			&\leq \frac{2\lambda_{\mathrm L}}{\sqrt n}
			(\sqrt{m_n\overline\Gamma_n}+\sqrt{s_0\Gamma_0}).
		\end{align*}
		Thus every fixed linear form
		\begin{align*}
			Z_{F,z} &\coloneqq \langle a_{F,z}-a_0,\varepsilon\rangle
		\end{align*}
		has common sub-Gaussian scale bounded by
		\begin{align*}
			v_{\mathrm{out}}
			&\coloneqq \frac{2\sigma\lambda_{\mathrm L}}{\sqrt n}
			(\sqrt{m_n\overline\Gamma_n}+\sqrt{s_0\Gamma_0}).
		\end{align*}
		
		To apply Lemma~\ref{lem:rare-event-maximum-lasso-correction} on the nonlocalized event, we first bound the global candidate class. A support of size \(j\) has at most \(2^j\) sign vectors, and every attainable support has size at most \(m_n\). Therefore
		\begin{align*}
			|\mathcal C_n|
			&\leq \sum_{j=0}^{m_n}2^j\binom pj
			\leq 2^{m_n}\sum_{j=0}^{m_n}\binom pj
			\leq 2^{m_n}\left(\frac{ep}{m_n}\right)^{m_n},
		\end{align*}
		where the last inequality holds for \(1\leq m_n\leq p\). Consequently,
		\begin{align*}
			\log(2|\mathcal C_n|)
			&\leq \log 2+m_n\log 2+m_n\log\left(\frac{ep}{m_n}\right)
			\lesssim m_n\log\left(\frac{ep}{m_n}\right).
		\end{align*}
		
		If \(\delta_n(r)=0\), then \(G_n(r)^c\) is a null event and the nonlocalized contribution is zero, so assume \(\delta_n(r)>0\).
		Using this lemma with \(B=G_n(r)^c\), \(\alpha=\delta_n(r)\), \(N=|\mathcal C_n|\), and \(v=v_{\mathrm{out}}\) gives
		\begin{align*}
			|\mathbb E[Y\mathbf1_{G_n(r)^c}]|
			&\leq \mathbb E[\mathbf1_{G_n(r)^c}\max_{(F,z)\in\mathcal C_n}|Z_{F,z}|] \\
			& \lesssim v_{\mathrm{out}}\,\delta_n(r)
			\sqrt{\log\left(\frac{2|\mathcal C_n|}{\delta_n(r)}\right)}.
		\end{align*}
		Using
		\begin{align*}
			\log\left(\frac{2|\mathcal C_n|}{\delta_n(r)}\right)
			&= \log(2|\mathcal C_n|)+\log\left(\frac1{\delta_n(r)}\right) \\
			& \lesssim m_n\log\left(\frac{ep}{m_n}\right)+\log\left(\frac1{\delta_n(r)}\right),
		\end{align*}
		and substituting the common scale \(v_{\mathrm{out}}\), we obtain
		\begin{align*}
			|\mathbb E[Y\mathbf1_{G_n(r)^c}]|
			& \lesssim \frac{\sigma\lambda_{\mathrm L}}{\sqrt n}
			(\sqrt{m_n\overline\Gamma_n}+\sqrt{s_0\Gamma_0})\delta_n(r) \\
			& \qquad\times
			\sqrt{m_n\log\left(\frac{ep}{m_n}\right)+\log\left(\frac1{\delta_n(r)}\right)} \\
			&= \mathcal N_n(r).
		\end{align*}
		
		Combining the localized and nonlocalized bounds, recalling the definition of \(Y\), and taking the infimum over \(r\in\mathcal R_n\) gives
		\begin{align*}
			\left|\mathbb E\left[\frac{2}{n}\langle X_E\hat{\delta}_{A,E},\varepsilon\rangle\right]\right|
			& \lesssim \inf_{r\in\mathcal R_n}\left\{\mathcal L_n(r)+\mathcal N_n(r)\right\}.
		\end{align*}
	\end{proof}
	
	\paragraph{Connection to expected prediction improvement.}
	Combining the decomposition in Section~\ref{sec:prediction-improvement} with Proposition~\ref{prop:localized-lasso-correction-expectation} gives
	\begin{align*}
		\mathbb E\Delta(\hat{\beta}_{\mathrm L},\hat{\beta}_A)
		&\geq \mathbb E[H_{E,s}(A)]
		- C\inf_{r\in\mathcal R_n}\left\{\mathcal L_n(r)+\mathcal N_n(r)\right\},
	\end{align*}
	for a universal constant \(C>0\). Consequently, any regime in which the expected positive gain dominates the localization remainder yields strictly smaller expected empirical prediction error for the generalized estimator than for the Lasso. Proposition~\ref{prop:sufficient-penalty-conditions} supplies explicit lower bounds for the gain, while the results below give interpretable conditions under which the remainder is small.
	
	\begin{remark}[Balance of gain and localization remainder]
		The expected comparison is favorable whenever, for some localization radius \(r\), the combined remainder
		\(\mathcal L_n(r)+\mathcal N_n(r)\) is small relative to the expected positive gain. To illustrate the local scale, suppose that \(r=O(s_0)\), \(k_n(r)=O(s_0)\), the relevant amplification factors are bounded, and the probability logarithms do not dominate the corresponding model complexity terms. If, in addition, the expected positive gain is comparable to \(\lambda_{\mathrm L}^2s_0\), then
		\begin{align*}
			\mathcal L_n(r)
			&\lesssim
			\frac{\sigma\lambda_{\mathrm L}}{\sqrt n}
			s_0 q_n(r)
			\sqrt{\log\left(\frac{ep}{s_0}\right)}.
		\end{align*}
		At the usual scale
		\(\lambda_{\mathrm L}\asymp\sigma\sqrt{(\log p)/n}\),
		the ratio of this local contribution to the positive gain is therefore of order
		\begin{align*}
			q_n(r)
			\sqrt{\frac{\log(ep/s_0)}{\log p}}.
		\end{align*}
		The nonlocalized term has a different role: because localization is lost on \(G_n(r)^c\), it involves the global model complexity and requires sufficiently strong control of \(\delta_n(r)\). In particular, \(\delta_n(r)\to0\) alone need not make this contribution negligible. See Section~\ref{subsec:literature-support} for Lasso results supporting such localization and signed model stability.
	\end{remark}

	\subsection{Interpretation of local and global amplification}
	\label{subsec:spectral-interpretation}
	
	For \(1\leq m\leq p\), define the sparse minimum eigenvalue
	\begin{align*}
		\phi_{\min,n}(m)
		&\coloneqq
		\min_{\substack{F\subseteq[p] \\1\leq |F|\leq m}}
		\lambda_{\min}\left(\frac1nX_F^\top X_F\right).
	\end{align*}
	Every attainable support in the localization window satisfies
	\begin{align*}
		|F| &\leq k_n(r).
	\end{align*}
	Hence, whenever \(k_n(r)\geq1\) and
	\begin{align*}
		\phi_{\min,n}(k_n(r)) &>0,
	\end{align*}
	Lemma~\ref{lem:correction-vector-norm-bound} gives
	\begin{align*}
		\Gamma_n(r)
		&\leq \phi_{\min,n}(k_n(r))^{-1}.
	\end{align*}
	Similarly, if \(s_0\geq1\) and \(\phi_{\min,n}(s_0)>0\), then
	\begin{align*}
		\Gamma_0
		&\leq \phi_{\min,n}(s_0)^{-1},
	\end{align*}
	while \(\Gamma_0=0\) when \(s_0=0\). Thus the finite sample proposition does not require a particular sparse eigenvalue condition; such conditions are needed only when one wishes to obtain more explicit bounds on the local amplification factors.
	
	\begin{remark}[Random design interpretation]
		The local sparse eigenvalue conditions above are standard in many random design settings. In particular, Gaussian and subgaussian designs admit restricted eigenvalue or approximate isometry bounds uniformly over sufficiently sparse coordinate sets with high probability; see, for example, \citet{baraniuk2008simple} and \citet{raskutti2010restricted}. Thus, when localization gives \(k_n(r_n)=O(s_0)\) and \(s_0\) remains within a sparsity regime for which such uniform eigenvalue bounds hold, bounded local amplification \(\Gamma_n(r_n)=O(1)\) is compatible with familiar sparse random design settings.
	\end{remark}
	
	The global factor \(\overline\Gamma_n\) has a different role. It enters only on the complement of the localization event, where the proof no longer restricts the realized support to a neighborhood of \(S_0\) and therefore maximizes over the full attainable class. Proposition~\ref{prop:localized-lasso-correction-expectation} does not assume that \(\overline\Gamma_n\) is uniformly bounded. For regularized corrections, however, the quadratic penalty can itself provide such control.
	
	\begin{remark}[Global control from quadratic regularization]
		Suppose that there exists \(\rho>0\) such that
		\begin{align*}
			A_F^\top A_F &\succeq \rho I_{|F|}
		\end{align*}
		for every nonempty \(F\in\mathfrak F_n\cup\{S_0\}\). Writing
		\begin{align*}
			C_F &\coloneqq \Sigma_{n,F}+A_F^\top A_F,
		\end{align*}
		we have \(0\preceq\Sigma_{n,F}\preceq C_F\) and \(C_F\succeq\rho I_{|F|}\). Hence
		\begin{align*}
			C_F^{-1}\Sigma_{n,F}C_F^{-1}
			&\preceq C_F^{-1}
			\preceq \rho^{-1}I_{|F|},
		\end{align*}
		and therefore
		\begin{align*}
			\overline\Gamma_n &\leq \rho^{-1},
			\qquad
			\Gamma_0\leq\rho^{-1}
		\end{align*}
		when \(S_0\neq\varnothing\). For the isotropic choice
		\begin{align*}
			A_F^\top A_F &= \lambda_{\mathrm R}I_{|F|},
			\qquad \lambda_{\mathrm R}>0,
		\end{align*}
		diagonalizing \(\Sigma_{n,F}\) gives
		\begin{align*}
			\Gamma_F
			&=\max_i\frac{\mu_i}{(\mu_i+\lambda_{\mathrm R})^2}
			\leq\frac{1}{4\lambda_{\mathrm R}},
		\end{align*}
		where \(\mu_i\) are the eigenvalues of \(\Sigma_{n,F}\). Thus
		\begin{align*}
			\overline\Gamma_n &\leq\frac{1}{4\lambda_{\mathrm R}},
		\end{align*}
		independently of the smallest eigenvalue of the attainable Gram matrices. Least squares refitting, corresponding to \(A_F=0\), does not have this regularization based protection.
	\end{remark}
	
	For a general design, \(\delta_n(r_n)\to0\) alone does not imply \(\mathcal N_n(r_n)\to0\). A sufficient condition for the nonlocalized term to vanish is
	\begin{align*}
		\frac{\sigma^2\lambda_{\mathrm L}^2}{n}
		&(\sqrt{m_n\overline\Gamma_n}+\sqrt{s_0\Gamma_0})^2
		\delta_n(r_n)^2 \\
		&\times
		\left[
		m_n\log\left(\frac{ep}{m_n}\right)
		+\log\left(\frac1{\delta_n(r_n)}\right)
		\right]
		\longrightarrow0.
	\end{align*}
	Thus the finite sample bound remains valid without a global spectral lower bound or a uniform penalty bound. On the complement of the localization event, the localization failure probability must offset the combined effects of model complexity, tuning, and global amplification. When the quadratic penalty supplies uniform control as above, the conditioning component of this remainder is automatically bounded.

	\subsection{Supporting evidence from the Lasso localization literature}
	\label{subsec:literature-support}
	
	The localization quantities used above have close precedents in the Lasso literature. Results such as \citet{zhanghuang2008} and \citet{meinshausenyu2009} show, under sparse eigenvalue type conditions, that Lasso selected models or KKT equality sets can remain sparse with high probability. In exactly sparse settings, such results are compatible with localization windows of order
	\begin{align*}
		|E\triangle S_0| &= O(s_0).
	\end{align*}
	Stronger model selection results provide related evidence for signed model stability. Here, sign consistency refers to recovery of the full coefficient sign vector, while variable selection consistency refers only to recovery of the support. Under additional incoherence, signal strength, and tuning conditions, results such as \citet{zhao2006model} and \citet{wainwright2009sharp} establish regimes in which the Lasso recovers the true signed support with probability tending to one. See also \citet[Chapter~11]{hastie2015statistical} for a broader discussion of Lasso support recovery and variable selection consistency. In settings where the Lasso support coincides with its equicorrelation set, as occurs almost surely under general position with an absolutely continuous response law, signed support recovery gives
	\begin{align*}
		(E,s) &= (S_0,z_0)
	\end{align*}
	on the recovery event. Such results therefore provide relevant examples of regimes in which both localization failure and signed mismatch are rare. Stronger recovery analyses can also provide quantitative control of recovery failure probabilities under additional assumptions, but no particular rate is assumed here.
	
	The cited literature is therefore used only to indicate established regimes relevant to the localization quantities in Proposition~\ref{prop:localized-lasso-correction-expectation}. The proposition itself is a finite sample result stated directly in terms of \(r_n\), \(\delta_n(r_n)\), \(q_n(r_n)\), and the spectral quantities, and does not require exact support or sign recovery.
	
	\section{Conclusion}
	
	The generalized quadratic correction provides a common representation for several Lasso based refitting procedures and makes the effect of the second stage penalty explicit through \(A^\top A\). Prediction improvement over the Lasso decomposes into a positive gain term and a selection dependent noise term; the localization bound controls the latter while keeping model complexity and conditioning visible. The resulting framework therefore identifies concrete conditions under which a structured correction can improve prediction without making exact support recovery a standing assumption, while also suggesting new refitting approaches through alternative choices of the quadratic penalty.
	
	\clearpage
	\appendix
	
	\section{Derivation of the Prediction Improvement Decomposition}
	\label{app:prediction-decomposition}
	
	The decomposition used in Section~\ref{sec:prediction-improvement} follows the same projection argument as in the Lasso--Ridge analysis of \citet{liu2026prediction}. We record the generalized version here for completeness.
	
	If \(E=\varnothing\), then \(\hat{\delta}_A=0\), and hence
	\begin{align*}
		\Delta(\hat{\beta}_{\mathrm L},\hat{\beta}_A)=0.
	\end{align*}
	Under the zero dimensional conventions adopted above, the right-hand side of~\eqref{eq:prediction-decomposition} is also zero. We therefore assume \(E\neq\varnothing\) for the derivation. Starting from the definition of prediction improvement and using \(\hat{\beta}_A=\hat{\beta}_{\mathrm L}+\hat{\delta}_A\), we have
	\begin{align}
		\Delta(\hat{\beta}_{\mathrm L},\hat{\beta}_A)
		&=\frac{1}{n}\|X\beta_0-X\hat{\beta}_{\mathrm L}\|_2^2
		-\frac{1}{n}\|X\beta_0-X\hat{\beta}_{\mathrm L}-X\hat{\delta}_A\|_2^2 \notag\\
		&=\frac{2}{n}\langle X\beta_0-X\hat{\beta}_{\mathrm L},X\hat{\delta}_A\rangle
		-\frac{1}{n}\langle X\hat{\delta}_A,X\hat{\delta}_A\rangle \notag\\
		&=\frac{2}{n}\langle X\beta_0-X_E\hat{\beta}_{\mathrm L,E},X_E\hat{\delta}_{A,E}\rangle
		-\frac{1}{n}\langle X_E\hat{\delta}_{A,E},X_E\hat{\delta}_{A,E}\rangle.
		\label{eq:app-prediction-basic-expansion}
	\end{align}
	By Proposition~\ref{prop:lasso-closed-form},
	\begin{align*}
		\hat{\beta}_{\mathrm L,E}
		=(X_E^\top X_E)^{-1}(X_E^\top y-n\lambda_{\mathrm L}s).
	\end{align*}
	Therefore, using \(y=X\beta_0+\varepsilon\),
	\begin{align*}
		X\beta_0-X_E\hat{\beta}_{\mathrm L,E}
		&=X\beta_0-X_E(X_E^\top X_E)^{-1}(X_E^\top y-n\lambda_{\mathrm L}s)\\
		&=X\beta_0-X_E(X_E^\top X_E)^{-1}
		\bigl(X_E^\top(X\beta_0+\varepsilon)-n\lambda_{\mathrm L}s\bigr)\\
		&=P_{X_E}^{\perp}X\beta_0-P_{X_E}\varepsilon
		+n\lambda_{\mathrm L}X_E(X_E^\top X_E)^{-1}s,
	\end{align*}
	where
	\begin{align*}
		P_{X_E}\coloneqq X_E(X_E^\top X_E)^{-1}X_E^\top,
		\qquad
		P_{X_E}^{\perp}\coloneqq I_n-P_{X_E}.
	\end{align*}
	Substituting this representation into~\eqref{eq:app-prediction-basic-expansion} gives
	\begin{align}
		\Delta(\hat{\beta}_{\mathrm L},\hat{\beta}_A)
		&=\frac{2}{n}\langle P_{X_E}^{\perp}X\beta_0,X_E\hat{\delta}_{A,E}\rangle
		-\frac{2}{n}\langle P_{X_E}\varepsilon,X_E\hat{\delta}_{A,E}\rangle \notag\\
		&\quad
		+\frac{2}{n}\langle n\lambda_{\mathrm L}X_E(X_E^\top X_E)^{-1}s,
		X_E\hat{\delta}_{A,E}\rangle
		-\frac{1}{n}\langle X_E\hat{\delta}_{A,E},X_E\hat{\delta}_{A,E}\rangle.
		\label{eq:app-prediction-projection-expansion}
	\end{align}
	The first inner product vanishes because \(P_{X_E}^{\perp}X_E=0\). Moreover, since \(X_E\hat{\delta}_{A,E}\in\operatorname{col}(X_E)\),
	\begin{align*}
		\langle P_{X_E}\varepsilon,X_E\hat{\delta}_{A,E}\rangle
		&=\langle \varepsilon,X_E\hat{\delta}_{A,E}\rangle,
	\end{align*}
	and
	\begin{align*}
		\frac{2}{n}\langle n\lambda_{\mathrm L}X_E(X_E^\top X_E)^{-1}s,
		X_E\hat{\delta}_{A,E}\rangle
		&=2\lambda_{\mathrm L}\langle \hat{\delta}_{A,E},s\rangle.
	\end{align*}
	Consequently,
	\begin{align*}
		\Delta(\hat{\beta}_{\mathrm L},\hat{\beta}_A)
		&=2\lambda_{\mathrm L}\langle \hat{\delta}_{A,E},s\rangle
		-\frac{1}{n}\langle X_E\hat{\delta}_{A,E},X_E\hat{\delta}_{A,E}\rangle
		-\frac{2}{n}\langle X_E\hat{\delta}_{A,E},\varepsilon\rangle,
	\end{align*}
	as claimed in~\eqref{eq:prediction-decomposition}.
	
	\section{Proofs of Technical Lemmas}
	\label{app:technical-lemmas}
	This appendix collects three auxiliary ingredients used in the proof of Proposition~\ref{prop:localized-lasso-correction-expectation}. The first is a binomial sum bound, the second controls the maximum of finitely many Sub-Gaussian variables on an event of small probability, and the third gives deterministic norm and spectral bounds for the correction vectors.
	
	\begin{lemma}[Binomial sum bound]
		\label{lem:binomial-sum-bound}
		For every integer \(p\geq1\) and every integer \(1\leq r\leq p\),
		\begin{align*}
			\sum_{j=0}^r\binom{p}{j}
			&\leq \left(\frac{ep}{r}\right)^r.
		\end{align*}
	\end{lemma}
	
	\begin{proof}
		Set \(x=r/p\in(0,1]\). For every \(0\leq j\leq r\), we have \(x^j\geq x^r\). Hence, by the binomial theorem,
		\begin{align*}
			(1+x)^p
			=\sum_{j=0}^p\binom{p}{j}x^j
			\geq \sum_{j=0}^r\binom{p}{j}x^j
			\geq x^r\sum_{j=0}^r\binom{p}{j}.
		\end{align*}
		Therefore,
		\begin{align*}
			\sum_{j=0}^r\binom{p}{j}
			&\leq x^{-r}(1+x)^p \\
			&=\left(\frac{p}{r}\right)^r\left(1+\frac{r}{p}\right)^p \\
			&\leq \left(\frac{p}{r}\right)^r e^r
			=\left(\frac{ep}{r}\right)^r,
		\end{align*}
		where the last inequality uses \(1+u\leq e^u\) for \(u\geq0\).
	\end{proof}
	
	\begin{lemma}[Maximum of Sub-Gaussian variables on an event of small probability]
		\label{lem:rare-event-maximum-lasso-correction}
		Let \(N\geq1\) and \(v>0\), and suppose that
		\(Z_1,\ldots,Z_N\) satisfy
		\begin{align*}
			\mathbb P(|Z_j|>t) &\leq 2\exp\left(-\frac{t^2}{2v^2}\right), \qquad t\geq0,
		\end{align*}
		for every \(j=1,\ldots,N\). Let
		\begin{align*}
			Z_{\max} &\coloneqq \max_{1\leq j\leq N}|Z_j|.
		\end{align*}
		If \(B\) is an event satisfying
		\begin{align*}
			\mathbb P(B) &\leq \alpha, \qquad \alpha\in[0,1],
		\end{align*}
		then
		\begin{align*}
			\mathbb E[ Z_{\max}\mathbf1_B ] & \lesssim v\alpha \sqrt{ \log\left(\frac{2N}{\alpha}\right) },
		\end{align*}
		where \(v\alpha\sqrt{\log(2N/\alpha)}\) is interpreted as zero when \(\alpha=0\). No independence among \(Z_1,\ldots,Z_N\), and no
		independence between \(B\) and the \(Z_j\), is required.
	\end{lemma}
	
	\begin{proof}
		The union bound gives
		\begin{align*}
			\mathbb P(Z_{\max}>t) &\leq 2N \exp\left(-\frac{t^2}{2v^2}\right).
		\end{align*}
		For \(\alpha>0\),
		\begin{align*}
			\mathbb E[ Z_{\max}\mathbf1_B ] &= \int_0^\infty \mathbb P(B\cap\{Z_{\max}>t\})\,dt \\
			&\leq \int_0^\infty \min\left\{\alpha,\, 2N\exp\left(-\frac{t^2}{2v^2}\right)\right\}\,dt.
		\end{align*}
		Set
		\begin{align*}
			t_\alpha &\coloneqq v \sqrt{ 2\log(2N/\alpha) }.
		\end{align*}
		Then
		\begin{align*}
			\mathbb E[ Z_{\max}\mathbf1_B ] &\leq \alpha t_\alpha + 2N \int_{t_\alpha}^\infty \exp\left(-\frac{t^2}{2v^2}\right)\,dt \\
			&\leq \alpha t_\alpha + \frac{\alpha v^2}{t_\alpha} \\
			& \lesssim v\alpha \sqrt{ \log(2N/\alpha) }.
		\end{align*}
	\end{proof}
	\begin{lemma}[Norm and spectral bounds for correction vectors]
		\label{lem:correction-vector-norm-bound}
		For every attainable signed model \((F,z)\in\mathcal C_n\),
		\begin{align*}
			\|a_{F,z}\|_2 &\leq \frac{2\lambda_{\mathrm L}}{\sqrt n}\sqrt{|F|\Gamma_F}.
		\end{align*}
		Moreover,
		\begin{align*}
			\|a_0\|_2 &\leq \frac{2\lambda_{\mathrm L}}{\sqrt n}\sqrt{s_0\Gamma_0},
		\end{align*}
		and, for every nonempty \(F\in\mathfrak F_n\cup\{S_0\}\) such that \(\Sigma_{n,F}\succ0\),
		\begin{align*}
			\Gamma_F
			&\leq \frac{1}{\lambda_{\min}\!(\Sigma_{n,F}+A_F^\top A_F)}
			\leq \frac{1}{\lambda_{\min}(\Sigma_{n,F})}.
		\end{align*}
	\end{lemma}
	
	\begin{proof}
		The bound is immediate for \((F,z)=(\varnothing,\varnothing)\), since \(a_{\varnothing,\varnothing}=0\) and \(\Gamma_\varnothing=0\). For every nonempty attainable signed model \((F,z)\in\mathcal C_n\),
		\begin{align*}
			\|a_{F,z}\|_2^2
			&= \frac{4\lambda_{\mathrm L}^2}{n}z^\top
			(\Sigma_{n,F}+A_F^\top A_F)^{-1}\Sigma_{n,F}
			(\Sigma_{n,F}+A_F^\top A_F)^{-1}z \\
			&\leq \frac{4\lambda_{\mathrm L}^2}{n}\Gamma_F\|z\|_2^2
			= \frac{4\lambda_{\mathrm L}^2}{n}|F|\Gamma_F.
		\end{align*}
		The bound for \(a_0\) follows by the same argument when \(S_0\neq\varnothing\), and is immediate from the conventions above when \(S_0=\varnothing\). For the spectral bound, let \(F\in\mathfrak F_n\cup\{S_0\}\) be nonempty with \(\Sigma_{n,F}\succ0\). Since \(A_F^\top A_F\succeq0\),
		\begin{align*}
			0\preceq\Sigma_{n,F}\preceq\Sigma_{n,F}+A_F^\top A_F.
		\end{align*}
		Because \(\Sigma_{n,F}+A_F^\top A_F\succ0\), congruence by
		\((\Sigma_{n,F}+A_F^\top A_F)^{-1}\) preserves the Loewner order and gives
		\begin{align*}
			0
			&\preceq
			(\Sigma_{n,F}+A_F^\top A_F)^{-1}
			\Sigma_{n,F}
			(\Sigma_{n,F}+A_F^\top A_F)^{-1} \\
			&\preceq
			(\Sigma_{n,F}+A_F^\top A_F)^{-1}.
		\end{align*}
		Therefore,
		\begin{align*}
			\Gamma_F
			&\leq
			\|
			(\Sigma_{n,F}+A_F^\top A_F)^{-1}
			\|_2 \\
			&=
			\frac{1}{\lambda_{\min}\!(\Sigma_{n,F}+A_F^\top A_F)}.
		\end{align*}
		Moreover,
		\begin{align*}
			\Sigma_{n,F}+A_F^\top A_F\succeq\Sigma_{n,F}
		\end{align*}
		implies
		\begin{align*}
			\lambda_{\min}\!(\Sigma_{n,F}+A_F^\top A_F)
			&\geq\lambda_{\min}(\Sigma_{n,F}),
		\end{align*}
		and hence
		\begin{align*}
			\Gamma_F
			&\leq
			\frac{1}{\lambda_{\min}\!(\Sigma_{n,F}+A_F^\top A_F)}
			\leq
			\frac{1}{\lambda_{\min}(\Sigma_{n,F})}.
		\end{align*}
	\end{proof}
	The final spectral inequality shows that lower bounds on sparse minimum eigenvalues remain convenient sufficient conditions for controlling the amplification factors. The quantity \(\Gamma_F\), however, is more targeted: it measures only the amplification that enters the generalized correction.
	
	\bibliographystyle{plainnat}
	\bibliography{references}
\end{document}